\documentclass[runningheads]{llncs}

\usepackage{times}
\usepackage{soul}
\usepackage{url}
\usepackage[hidelinks]{hyperref}
\usepackage[utf8]{inputenc}
\usepackage[small]{caption}
\usepackage{graphicx}
\usepackage{amsmath}
\usepackage{booktabs}
\usepackage{algorithm}
\usepackage{algorithmic}
\usepackage{marvosym}

\usepackage{amsthm}

\begin{document}
\title{Committee Elections with Candidate Attribute Constraints}
%
%
\author{Aizhong Zhou\inst{1} \textsuperscript{(\Letter)} \and
Fengbo Wang\inst{2} \and
Jiong Guo\inst{3}}
\authorrunning{A. Zhou et al.}
%
\institute{
Computer Science and Technology, Ocean University of China
\email{{zhouaizhong}@ouc.edu.cn}\\
\and 
Computer Science and Technology, Ocean University of China\\
\email{{wfb}@stu.ouc.edu.cn}\\
\and
Computer Science and Technology, Shandong University\\
\email{JGuo@sdu.edu.cn}}
\maketitle              
%


\begin{abstract}
In many real-world applications of committee elections, the candidates are associated with certain attributes and the chosen committee is required to satisfy some constraints posed on the candidate attributes.
For instance, when dress collocation, it is generally acknowledged that when wearing a tie, you'd better wear a shirt, and wearing a suit, you'd better wear leather shoes.
Here, dresses are categorized by upper garment, lower garment, shoes et.al, and upper garment is with the attribute tie and shirt, lower garment is with the attribute suit, and shoes is with the attribute leather.
And two constraints ``tie infers shirt'' and ``suit infers leather shoes'' are proposed.
We study this variant of committee elections from the computational complexity viewpoint.
Given a set of candidates, each with some attributes and a profit, and a set of constraints, given as propositional logical expressions of the attributes, the task is to compute a set of $k$ candidates, whose attributes satisfy all constraints and whose total profit achieves a given bound.
We achieve a dichotomy concerning classical complexity with no length limit on constraints: the problem is polynomial-time solvable, if the following two conditions are fulfilled: 1) each candidate has only one attribute and 2) each attribute occurs at most once in the constraints.
It becomes NP-hard if one of the two conditions is violated.
Moreover, we examine its parameterized complexity.
The parameterization with the number of constraints, the size of the committee, or the total profit bound as parameter leads to para-NP-hardness or W[1]-hardness, while with the number of attributes or the number of candidates as parameter, the problem turns out to be fixed-parameter tractable.
Because of the para-NP-hardness result that only with one complex constraint, this problem is NP-hard, we continue to study the complexity of this problem where each constraint contains only few (two or three) attributes, and receive a dichotomous classical complexities.
\end{abstract}

\section{Introduction}
The problem of aggregating the preferences of different agents (voters) occurs in diverse situations in our real life.
This problem also plays a fundamental role in artificial intelligence and social choice~\cite{FVUJA,C-CACM-2010}.
While most cases studied are set to find out a single winner, voting can also be used to select a fixed-size set of winners (multiwinner), called committee.
Recently, committee election has been extensively studied from the axiomatic and algorithmic aspects~\cite{GNJEB,HSJSNT,ABCEFW-SCW-2017}.
Multi-winner voting rules have received a considerable amount of attention, due to their
significant applications in social choice~\cite{DKNS01,LB11,LB15,SFL-AI-2016}.

The problem of finding a committee with diversity constraints has also been formally studied in social choice.
Yang~\cite{YW-IJCAI-2018} and Kilgour~\cite{K-TD-2016} studied the multiwinner voting with restricted admissible sets where given a graph whose vertex set is the candidate set, the winners selected should satisfy some combinatorial restrictions in the graph, such as connected, fixed size, or independent.
Talmon~\cite{T-TCS-2018} and Igarashi et al.~\cite{TPE-AAAI-2017} studied the election by considering the relations between voters where voters are represented as vertices in a graph or a social network, and the winners should satisfy some combinatorial restrictions.
Furthermore, in several works, the problem of diverse committee selection is formulated as an election with candidates restrictions where candidates with attributes.
Aziz et al.~\cite{ALM-IJCAI-2016,AM-AAMAS-2020}, Lang and Skowron \cite{LS-AI-2018}, and Bredereck et al.~\cite{BFILS-AAAI-2018} studied the models seeking for a committee with quantity restriction, that for each attribute offers a certain representation or the occurrence of the attribute needs to be in the interval.
Celis et al. \cite{CHV-IJCAI-2018} introduced an algorithmic framework for multiwinner voting problems where there is an additional requirement that the selected committee should be fair with respect to a given set of attributes. 
They also presented several approximation algorithms and matching approximation hardness results for various attribute group structures and types of score functions (such as monotone and submodular score functions).
Recently, Aziz \cite{A-GDN-2019} proposed a cubic-time algorithm for committee election with soft diversity constraints, considering two axioms called type optimality and justified envy-freeness.

In this paper, we continue this line of research, that is, studying committee elections with input consisting of a set of candidates, each candidate associated with a profit and a set of attributes, and a set of attribute constraints.
We complement the previous work by proposing a model seeking for a committee, which maximizes the total profit and satisfies all constraints, which are defined as propositional logic expressions of attributes.
We call this new model Committee Election with Candidate Attribute Constraints (CECAC).
We analyze the classical and parameterized complexity of this model and obtain the following dichotomy concerning its classical complexity: If $1)$ each candidate has only one attribute and $2)$ each attribute occurs at most once in the constraints, then CECAC can be solved in polynomial time.
It becomes NP-hard, if one of $1)$ and $2)$ is violated.
Concerning parameterized complexity, we consider five parameters: the size of committee, the total profit bound, the number of attributes, the number of candidates, and the number of constraints.
The number of candidates is one of the standard parameters while studying parameterized complexity of election problems~\cite{BBCN-MARB-2012,BCFGNW-2014}; most committee election problems have been examined concerning the parameterization of the size of the committee~\cite{AFGST-IJCAI-2018,GJRSZ-AAMAS-2019,LG-AAMAS-2016,MNS-AAMAS-2015}.
The other three parameters are more problem-specific.
It is conceivable that most practical instances of CECAC admit relatively small values for the number of attributes and the number of constraints.
In addition, we continue to study the CECAC problem where the number of attributes occurring in each constraint is small, we call it as \emph{Simple Constraints}, and receive a dichotomous results for classical complexity.
Our results are summarized in Table 1.

\begin{table*}[]
\label{tab:constraints results}
\begin{tabular}{c|c|c|c|c}
\hline
{} & \multicolumn{2}{|c|}{Classical Complexity} & \multicolumn{2}{c}{Parameterized Complexity}\\
\hline
{} & {Conditions} & {Complexity} & {Parameter} & {Complexity}\\
\hline
{General Constraints} & {$|\alpha(c)|\leq 1$ and $N(a)\leq 1$} & {P} & {m} & {FPT}\\
\cline{2-5}
{} & {$|\alpha(c)|> 1$ or $N(a)\leq 1$} & {NP-h} & {$\ell$} & {FPT}\\
\hline
{Simple Constraints} & {$|\alpha(c)|\leq 1$, $N(a)\leq 2$ and {$N(r)\leq 2$}} & {P} & {d} & {para-NP-h}\\
\cline{2-5}
{} & {$|\alpha(c)|> 1$, $N(a)> 2$ or {$N(r)> 2$}} & {NP-h} & {$k$ and $p$} & {W[1]-h}\\
\hline
\end{tabular}
\centering\caption{Classiacl complexity of Committee Election with Candidates Attribute Constraints (CECAC) with simple constraints. CECAC is polynomial-time solvable if each candidate has at most one attribute $(|\alpha(c)|\leq 1)$, each attribute occurs at most 2 constraints $(N(a)\leq 2)$, and each simple constraint has at most two attribute $(L(r)\leq 2)$; otherwise, CECAC with simple constraints is NP-hard. 
}
\end{table*}

Our model shares certain common features with the previous works.
For instance, as in the model by Lang and Skowron~\cite{LS-AI-2018}, we completely ignore the votes (the support of each candidate from all votes is shown as the profit of this candidate) and focus on the influence of attributes and constraints on committee elections.
A majority of the previous work such as Lang and Skowron~\cite{LS-AI-2018}, Aziz et al.~\cite{ALM-IJCAI-2016}, Bredereck et al.~\cite{BFILS-AAAI-2018}, and Celis et al.~\cite{CHV-IJCAI-2018} considered the constraints which are defined independently for single attributes, for example, the number of committee members with a certain attribute.
In contrast, our model aims to introduce the logical relation of the attributes to the committee election problem.
Recently, Aziz~\cite{A-GDN-2019} considered the relation between the occurrences of attributes and the aim of their model is to satisfy as many constraints as possible.
Our model needs to select a committee satisfying all constraints, and focuses on the logical relation of attributes instead of quantity requirements posed on attributes.

\section{Preliminaries}~\label{sect:prel}
A committee election with candidate attribute constraints is denoted as a tuple~$(E=(C, A, R, S \alpha), k, p)$, where~$C$ is the set of candidates $C=\{c_1,c_2,\dots,c_m\}$, $A$ is the set of attributes $A=\{a_1, a_2,\dots, a_{\ell}\}$, $R$ is the set of constraints $R=\{r_1, \dots, r_d\}$, $S$ is the profit set that each candidate has a profit $S(c)\in S$, and $\alpha$ is an attribute function mapping the candidates to subsets of attributes, that is, $\alpha: C \rightarrow 2^{A}$, $k$ is the size of committee, $p$ is the profit bound of the committee.
We use $\alpha(c_i)$ to denote the set of attributes that $c_i$ owns and represent a candidate $c_i$ with $\alpha(c_i)=\{a_{i_1},\cdots, a_{i_{|\alpha(c_i)|}}\}$ as $c_i<a_{i_1},\cdots,a_{i_{|\alpha(c_i)|}}>$.
Let $S(c_i)$ be the profit of candidate $c_i$, and $S(C)$ be the total profit of candidates in $C$, $S(C)=\sum_{c_i \in C}S(c_i)$, if $C=\emptyset$, $S(C)=0$.
A subset of candidates $W \subset C$ with $|W|=k$ is called a $k$-committee.
Each constraint $r\in R$ is of the form $x_{1} \ \sigma_1 \ x_{2} \ \sigma_2 \ \cdots \ x_{t_1} \rightarrow x'_{1} \ \sigma'_1 \  x'_{2} \ \sigma'_2 \ \cdots \ x'_{t_2}$, where $\sigma_i, \sigma'_i \in \{\vee, \wedge\}$ and $x_i, x'_i$ are of the form $a_j$ or $\overline{a_j}$ for some attribute $a_j\in A$, for example, $r: a_1 \rightarrow \overline{a_2} \vee a_3$.
For the ease of presentation, we also use the notation $a_i$ to denote the Boolean variable, whose value is ``set" by a $k$-committee $W$, meaning that if one of the candidates in the committee owns the attribute $a_i$, then the corresponding variable $a_i$ is set to $1$ or TRUE; otherwise, $a_i$ is set to $0$ or FALSE. 
Accordingly, a $k$-committee  $W$ satisfying a constraint $r\in R$ means that the assignment of the Boolean variables in $r$ by $W$ evaluates $r$ to $1$ or TRUE.
For example, if $W$ satisfies $r: a_1 \rightarrow \overline{a_2} \vee a_3$, then $W$ does not contain a candidate $c$ with $a_1 \in \alpha(c)$, $W$ contains a candidate $c$ with $a_3\in \alpha(c)$, or $W$ does not contain a candidate with $a_2\in \alpha(c)$.

\subsection{Problem Definition}
We define the central problem of this paper.
\begin{quote}
    {\bf Committee Election with Candidate Attribute Constraints (CECAC)}\\
    {\bf Input}: An election $(E=(C, A, R, S, \alpha), k, p)$, each candidate $c\in C$ with some attributes $\alpha(c)\in  A$ and a profit $S(c)\in S$, a set of logical constraints $R$ of attributes, and two positive integers $k\leq |C|$ and $p$.\\
    {\bf Question}: Is there a committee $W \subseteq C$ with~$|W|=k$ satisfying
 all constraints in $R$ and $\sum_{c\in W}S(c)\geq p$?
\end{quote}

In this paper, we consider the classical as well as parameterized complexity of CECAC with respect to the following parameters: the number of candidates $m=|C|$, the number of attributes $\ell=|A|$, the number of constraints $d=|R|$, the size of committee $k=|W|$, and the total profit bound $p$.

Notice that there are two main difference between CECAC and the classical Satifiability and Constraint Satisfaction problems.
First, candidates associated with more than one attribute introduce certain assignment correlation of the variables corresponding to the attributes.
Second, the combination of candidates with different profit values and the bounded size of the committee causes a deliberate selection of candidates and thus more tricky assignment of variables in the constraints.

\subsection{Parameterized Complexity}
Parameterized complexity allows to give a more refined analysis of computational problems and in particular,
can provide a deep exploration of the connection between the problem complexity and various problem-specific parameters.
A fixed-parameter tracta-\\
ble (FPT) problem admits an $O(f(k) \cdot |I|^{O(1)})$-time algorithm, where~$I$ denotes the
whole input instance, $k$ is the parameter, and~$f$ can be any computable function.
Fixed-parameter intractable problems can be classified into many complexity classes, where the most fundamental ones are W[1] and W[2].
A problem is para-NP-hard with respect to parameter $k$, when the problem remains NP-hard even if $k$ is a  fixed constant.
For more details on parameterized complexity, we refer to ~\cite{CFKLMPPS15,DF99}.

\section{Classical Complexity}
In this section, we study the classical complexity of CECAC, presenting a dichotomy result classifying NP-hard and polynomial-time solvable cases with respect to the maximal number of attributes associated with a single candidate ${\max}_{c\in C}|\alpha(c)|$ and the maximal occurrence of an attribute in all constraints ${\max}_{a\in A}N(a)$, where $N(a)$ is the number of occurrence of $a$ or $\overline{a}$ in $R$.

\subsection{Polynomial-time solvable case}

\begin{theorem}
\label{A(c)-N(a)-p}
When the following two conditions are fulfilled, CECAC is solvable in polynomial time:\\
1) each candidate has at most one attribute, $\forall\ c \in C: |\alpha(c)|\leq 1$;\\
2) each attribute occurs at most once, $\forall\ a\in A: N(a)\leq 1$.
\end{theorem}
\begin{proof}
Let $(E=(C,A,R, S, \alpha), k,p)$ be a CECAC instance where $|C|=m$, $|A|=\ell$, and $|R|=d$.
We first make some modifications to the candidates in $C$, the attributes in $A$, and the constraints in $R$, such that each attribute belongs to at most one candidate.
Let $C(a_i)$ be the set of candidates with the attribute $a_i$ and $u_i=|C(a_i)|$.
We apply the following modification operations to each attribute $a_i$ with $u_i>1$:
\begin{itemize}
    \item Replacing $a_i$ by a set of new attributes $A_i=\{a^1_i, a^2_i, \cdots, a^{u_i}_i\}$,
    \item For each candidate $c<a_i>$, replacing $a_i$ by a unique attribute in $A_i$,
    \item In the constraint where $a_i$ or $\overline{a_i}$ occurs, replacing $a_i$ by $a^1_i \vee a^2_i \cdots \vee a^{u_i}_i$ and $\overline{a_i}$ by $\overline{a^1_i \vee a^2_i \cdots \vee a^{u_i}_i}$.
\end{itemize}
For example, there are three candidates with the attribute $a_1: c_1<a_1>, c_2<a_1>$,  $c_3<a_1>$, two candidates with the attribute $a_2: c_4<a_2>, c_5<a_2>$, and a constraint $r_1: a_1\rightarrow \overline{a_2}$.
We apply the above modification operations by creating three new attributes $a^1_1, a^2_1, a^3_1$, such that each of the three new attributes belongs to one of the three candidates:~$c_1<a^1_1>, c_2<a^2_1>, c_3<a^3_1>$.
Similarly, we create two attributes $a^1_2, a^2_2$, belonging to $c_4<a^1_2>$ and $c_5<a^2_2>$.
The constraint is changed to $r'_1: a^1_1 \vee a^2_1 \vee a^3_1 \rightarrow \overline{a_2^1 \vee a_2^2}$.
The two constraints $r_1$ and $r'_1$ express the same restriction that if one of the candidates $c_1, c_2, c_3$ is selected to the committee, then candidates $c_4$ and $c_5$ cannot be selected to the committee any more.
After the modification operation, we divide the candidates $c\in C$ into two parts $C^{+}$ and $C^{-}$: add a candidate $c_i<a_i>$ to $C^+$ if $a_i$ or $\overline{a_i}$ occurs in a constraint $r\in R$; otherwise, add $c_i<a_i>$ to $C^-$.
Obviously, choosing the candidates in $C^-$ or not has no influence on the satisfaction of constraints in $R$.

The main ideal of our algorithm consists of the following steps.
Firstly, we combine all constraints in $R$ into one propositional logical formula, denoted as $D(R)$, and construct a binary tree $T(D(R))$ of $D(R)$.
Then, according to the structure of $T(D(R))$, we use a bottom-up approach to select some candidates in $C^+$ to make sure $D(R)$ satisfied.
Next, we add some candidates from $C^-$ to $W$, or replace some selected candidates with the candidates in $C^-$ to achieve a $k$-committee $W$ with the highest total profit $S(W)$.
Finally, we decide whether there is a solution to the CECAC instance by comparing $S(W)$ with the total profit bound $p$.

In the following, we show the details of our algorithm.
Firstly, we combine all constraints in $R$ into one propositional logical formula $D(R)$.
For example, there are three constraints $r_1: a_1 \rightarrow a_2,\ r_2: a_3 \rightarrow a_4 \vee a_5,\ r_3: \overline{a_6} \rightarrow a_7$ in $R$.
We transform them into a formula: $ (\overline{a_1} \vee a_2) \wedge (\overline{a_3} \vee a_4 \vee a_5) \wedge ({a_6} \vee {a_7})$.
To illustrate the structure of $D(R)$, we construct a binary tree $T(D(R))$:
\begin{itemize}
    \item The root represents the whole formula $D(R)$,
    \item Each internal node is associated with a subformula of $D(R)$,
    \item The leaves represent the literals $a_i$ or $\overline{a_i}$.
\end{itemize}
The structure of $T(D(R))$ for the above example is shown in Figure~\ref{fig:binary tree}.
\begin{figure}
   \centering
    \center{\includegraphics[width=0.45\textwidth]{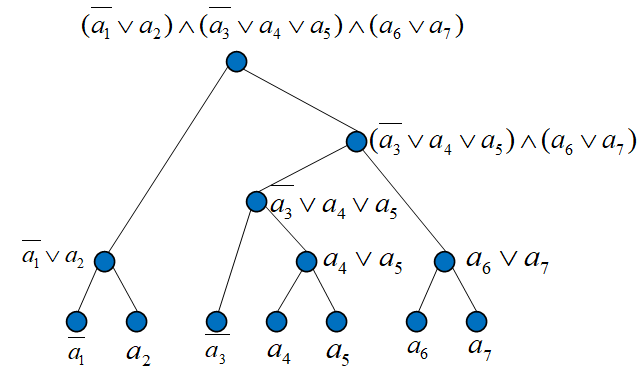}}  
    \caption{The root of the binary tree is set to $D(R)=(\overline{a_1}\vee a_2) \wedge (\overline{a_3} \vee a_4 \vee a_5) \wedge (a_6 \vee a_7)$, and each leaves is set to an attribute occuring in $D(R)$. In this binary tree, the subformulas $q$ and $q'$ of two children nodes can constitute the subformula of their father node by $q \vee q'$ or $q \wedge q'$.}
    \label{fig:binary tree}
\end{figure}
We use $Q^T(D(R))$ to denote the set of subformulas of $D(R)$ corresponding to the nodes in $T(D(R))$.
For the above example, $Q^T(D(R))=\{\overline{a_1}, a_2, \overline{a_3}, a_4, a_5, a_6, a_7, \overline{a_1}\vee a_2, a_6\vee a_7, a_4\vee a_5, \overline{a_3}\vee a_4 \vee a_5, (\overline{a_3}\vee a_4 \vee a_5)\wedge (a_6\vee a_7), (\overline{a_1}\vee a_2)\wedge (\overline{a_3}\vee a_4 \vee a_5)\wedge (a_6 \vee a_7) \}$.

Now, we show the details of the bottom-up approach.
Let $C(q)$ be the set of candidates whose attribute occurs in subformula $q$, $q\in Q^T(D(R))$.
For each internal node with a subformula $q$, we compute three multisets $V(q)=\bigcup_{0\leq j\leq i\leq k}\{V_i^j(q)\}$, $N(q)=\bigcup_{0\leq j\leq i\leq k}\{N_i^j(q)\}$ and $P(q)=\bigcup_{0\leq j\leq i\leq k}\{P_i^{i-j}(q)\}$.
Hereby, $N(q)$ and $P(q)$ are used to record the candidate set whose candidates cannot be replaced by the candidates in $C^-$ and the candidate set whose candidates can be replaced by the candidates in $C^-$, respectively.
The set $V_i^j(q)$ is set one-to-one corresponding to $N_i^j(q)$ and $P_i^{i-j}(q)$ that selecting total $i$ candidates with $j$ candidates in set $N_i^j(q)$ and $i-j$ candidates in $P_i^j(q)$.
The initialization of $N(q)$, $P(q)$ and $V(q)$ is as follows:
\begin{itemize}
    \item Set $N_i^j(q)=P_i^{i-j}(q)=\emptyset$ and $V_i^j(q)=-\infty$ for $0\leq j\leq i\leq k$ except for the following cases:\\
    $\textcircled{1}$ $q=a$ for $a\in A$: set $N_1^1(q)=\{c<a>\}$, $P_1^0(q)=\emptyset$, and $V_1^1(q)=S(c<a>)$;\\
    $\textcircled{2}$ $q=\overline{a}$ for $a\in A$: set $N_0^0(q)=P_0^0(q)=\emptyset$ and $V_0^0(q)=0$.
\end{itemize}
It is clear that $q=a$ is satisfied when $c<a>$ is chosen, and $c<a>$ is irreplaceable by the candidates in $C^-$.
Hereby, the highest total profit is clearly $S(c<a>)$.
Similarly, $q=\overline{a}$ is satisfied when $c<a>$ is not chosen and the highest total profit is 0.
Next, we show the details of how to satisfy subformula $q$ according to its two children $q', q''\in Q^T(D(R))$.
If $|C|=j$, let $T_j(C)$ be the subset of $C$ with the maximal total profit $S_j(C)$ and $|T_j(C)|=j$; otherwise, $T_j(C)=\emptyset$ and $S_j(C)=-\infty$.
\begin{itemize}
    \item (AND operation) If $q=q'\wedge q''$, then for each $0\leq j\leq i\leq k$, we set:
    \begin{equation*}
    \begin{split}
        V_i^j(q)=\max\{V_r^t(q')+V_{i-r}^{j-t}(q'')\ |\ 0\leq t
        \leq r\leq i\ {\tt and}\ 0\leq t\leq j\leq i\}
        \end{split}
    \end{equation*}
    Let $r, t$ be the combination achieving the maximal value $V_i^j(q)$.
    Set $N_i^j(q)=N_r^t(q')\cup N_{i-r}^{j-t}(q'')$ and $P_i^{i-j}(q)=P_r^{r-t}(q')\cup P_{i-r}^{(i-r)-(j-t)}(q'')$.
    \item (OR operation) If $q=q'\vee q''$, then for each $0\leq j\leq i\leq k$, we set:
$$
V_i^j(q)=\left\{
\begin{aligned}
&(1)\max\{V_r^t(q')+V_{i-r}^{j-t}(q'')\ |\ 0\leq t
        \leq r\leq i,\ 0\leq t\leq j\leq i\},& \\
&(2)\max\{V_r^j(q')+S_{i-r}(C(q''))\ |\ 0\leq j\leq r\leq i\},& \\
&(3)\max\{V_r^j(q'')+S_{i-r}(C(q'))\ |\ 0\leq j\leq r\leq i\}.& \\
\end{aligned}
\right.
$$
\begin{itemize}
    \item If the maximal value of $V_i^j(q)$ is achieved by (1) with some $t, r$, then set $N_i^j(q)=N_r^t(q')\cup N_{i-r}^{j-t}(q'')$ and $P_i^{i-j}(q)=P_r^{r-t}(q')\cup P_{i-r}^{(i-r)-(j-t)}(q'')$.
    \item If the maximal value of $V_i^j(q)$ is achieved by (2) with some $r$, then set $N_i^j(q)=N_r^j(q')$ and $P_i^{i-j}(q)=T_{i-j}(P_r^{r-j}(q')\cup C(q''))$.
    \item If the maximal value of $V_i^j(q)$ is achieved by (3) with some $r$, then set $N_i^j(q)=N_r^j(q'')$ and $P_i^{i-j}(q)=T_{i-j}(P_r^{r-j}(q'')\cup C(q'))$.
\end{itemize}
\end{itemize}

For AND operation with $q=q'\wedge q''$, both of $q'$ and $q''$ should be true in order to satisfy $q$.
We enumerate all combinations of $V_{r}^t(q')$ and $V_{i-r}^{j-t}(q'')$ with $0\leq t\leq r\leq i$ and $0\leq t\leq j \leq i$.
For OR operation with $q=q' \vee q''$, we consider three cases that both of $q'$ and $q''$, only $q'$, or only $q''$ is true.
$N_i^j(q)$ and $P_i^j(q)$ are set corresponding to the maximal value $V_i^j(q)$ of all cases.

By the above bottom-up approach, we achieve three sets: $V(D(R))$, $N(D(R))$, and $P(D(R))$ at the root.
For $N_i^j(D(R))$ and $P_i^{i-j}(D(R))$ with $i<k$, we add $k-i$ candidates in $C^-$ to $P_i^{i-j}(D(R))$ and achieve a $k$-committee.
Then replacing some candidates in $P_i^{i-j}(D(R))$ with the candidates in $C^-$ to achieve the highest total profit.
It means that the $k$-committee consists of $j$ candidates in $N_i^j(D(R))$ and $k-j$ candidates in $P_i^{i-j}(D(R))\cup C^-$.
Let $W$ be the $k$-committee with the highest total profit $S(W)=\max\{S(N_i^j(D(R)))+S_{k-j}(P_i^{i-j}(D(R))\cup C^-)\ |\ 0\leq j\leq i\leq k\}$. 
Finally, we decide whether there is a solution to the CECAC instance by comparing $S(W)$ with the total profit bound $p$.

Next, we prove the correctness and running time of the algorithm.
Concerning the running time, we make some modifications to the candidates in $C$, the attributes in $A$, and the constraints in $R$ to make sure that each attribute belongs to exactly one candidate.
The number of possible modification operations is bounded by $O(m\times d \times \ell)$.
The number of subformulas in $Q$ is bounded by $2 \times \ell$, since the number of nodes of a binary tree with $\ell$ leaves is at most $2 \times \ell$.
We compute the three sets $V^j_i(D(R)), N^j_i(D(R))$ and $P^j_i(D(R))$ by two types of operations.
The operations can be applied at most $O(\ell)$ times, one application for each internal node.
For each operation, the algorithm calculates $\frac{(k+1)(k+2)}{2}$ sets or values for each of  $V^j_i(D(R)), N^j_i(D(R))$ and $P^j_i(D(R))$ with $0 \leq j\leq i \leq k$, where each set or value can be achieved by at most $\frac{(k+1)(k+2)}{2}+2\times (k+1)$ comparisons.
So, the total number of comparisons is less than $O(\ell \times (k+1)^2\times(k+2)^2)$.
In addition, the value of $S(W)$ can be achieved by at most $\frac{(k+1)(k+2)}{2}$ comparisons.
In summary, the total running time is bounded by $O(k^4\times d \times \ell \times m)$.

Now, we turn to the correctness.
Obviously, the modifications to $C$, $A$, and $R$ have no affect on the output of the election, since the modified constraint set has the same restriction as the original one.
Let $C(q)$ be the set of candidates with attribute $a$ when $a$ or $\overline{a}$ occurs in $q$.
According to Conditions 1) and 2) of the theorem, choosing a candidate $c<a>$ or not is determined by the first subformula $q$ where $a$ or $\overline{a}$ occurs according to the bottom-up approach and the number of chosen candidates in $C(q)$.
The initialization of $V(q)$, $N(q)$ and $P(q)$ is clearly correct, since $q=a$ is satisfied when $c<a>$ is chosen with a total profit $S(c<a>)$ and no candidate can be replaced by the candidates in $C^-$, and $q=\overline{a}$ is satisfied when $c<a>$ is not chosen with a total profit 0. 
By AND and OR operations, we enumerate all satisfaction conditions of $q$.
And for each operation, we enumerate all combinations of candidates in $N(q)$ and $P(q)$.
In the following, we prove the correctness of AND and OR operations.
For the AND operation with $q=q'\wedge q''$, $q$ is true if and only if both of $q'$ and $q''$ are true.
For each $j$ and $i$ with $0\leq j\leq i\leq k$, we enumerate all combinations of $V^t_{r}(q')$ and $V^{j-t}_{i-r}(q'')$ with $0\leq t\leq r\leq i$ and $0\leq t\leq j\leq i$.
Each of the combination means that we choose $r$ candidates in $C(q')$ with $t$ candidates not replaceable and choose $i-r$ candidates in $C(q'')$ with $j-t$ candidates not replaceable.
We set $N^j_i(q)$ and $P^{i-j}_i(q)$ according to $V^j_i(q)$ where $V^{j}_i(q)=-\infty$ means $q$ cannot be satisfied by selecting exactly $i$ candidates in $C(q')$ including $j$ candidates which will not be replaced, and $V^{j}_i(q)\neq-\infty$ means there is a combination of $N^t_{r}(q')$, $P_r^{r-t}(q')$, $N_{i-r}^{j-t}(q'')$, and $P_{i-r}^{(i-r)-(j-t)}(q'')$ which can make $q$ true.
For the OR operation with $q=q'\vee q''$, $q$ can be true if and only if $q'$ is true or $q''$ is true.
We consider all three cases that both $q'$ and $q''$, only $q'$, or only $q''$ are satisfied.
For each $j$ and $i$ with $0\leq j\leq i\leq k$, we enumerate all $\frac{(k+1)(k+2)}{2}$ cases for both $q'$ and $q''$ being satisfied, and $2\times(k+1)$ cases for one of $q'$ or $q''$ being satisfied.
Similarly, $V^{j}_i(q)=-\infty$ means $q$ cannot be satisfied by choosing exactly $i$ candidates including $j$ irreplaceable candidates, and $V^{j}_i(q)\neq-\infty$ means there is a combination which can make $q$ being true.
For both of AND and OR operations, we compute $N(q)$, $P(q)$, and $V(q)$.
Finally, we arrive at three sets $V(D(R))$, $N(D(R))$, and $P(D(R))$ at the root.
We enumerate all $\frac{(k+1)(k+2)}{2}$ cases of solutions with $j$ $(0\leq j\leq i\leq k)$  candidates in $N_i^j(D(R))$ and $k-j$ candidates in $P_i^{i-j}(D(R))\cup C^-$, and achieve the highest total profit $S(W)$.
Therefore, if $S(W)<p$, there is no solution.
Otherwise, let $W$ be the set of candidates corresponding to $S(W)$.
Clearly, $W$ is a solution with $k$ candidates satisfying all constraints and its total profit is at least $p$.
This completes the proof of the correctness of the algorithm and thus, the proof of the theorem.
\end{proof}

\subsection{NP-hard cases}
\begin{theorem}
\label{A(c)-NP}
CECAC is NP-hard when each candidate has at most two attributes and each attribute occurs at most once in the constraints.
\end{theorem}
\begin{proof}
We prove the theorem by reducing {\sc Clique} on $D$-regular graphs to CECAC, which given a~$D$-regular graph~$\mathcal{G}=(\mathcal{V},\mathcal{E})$ and an integer~$k'$ asks for a size-$k'$ clique.
Each vertex in $D$-regular graph has the same degree $D$.
A clique~$\mathcal{K}$ in an undirected graph~$\mathcal{G}$ is a subset of vertices, which form a complete graph.
In other words, the vertices in the clique~$\mathcal{K}$ are pairwisely adjacent in~$\mathcal{G}$.
{\sc Clique} is NP-hard and W[1]-hard with respect to the size of clique~\cite{DF99}.
We construct a CECAC instance~$(E=(C,A,R, S, \alpha),k,p)$ from~$(\mathcal{G}=(\mathcal{V},\mathcal{E}), k')$ as follows.

For each edge $e_j \in \mathcal{E}$, we create four attributes in $X_j=\{x^1_j$, $x^2_j, x^3_j, x^4_j\}$.
For each vertex $v_i \in \mathcal{V}$, we create $2D+2$ attributes in $Y_i \cup Z_i$: $Y_i=\bigcup_{0 \leq h \leq D}\{y^h_i\}$,\ $Z_i=\bigcup_{0 \leq h \leq D}\{z^h_i\}$.
Attributes $y^h_i, 1\leq h\leq D$, one-to-one corresponds to the $D$ edges incident to vertex $v_i$.
Let $A=(\bigcup_{1\leq j\leq |\mathcal{E}|}X_j) \cup (\bigcup_{0\leq i\leq |\mathcal{V}|}(Y_i \cup Z_i))$.

In addition, we create four sets of candidates.
For each edge $e_j$, we create an edge candidate in $C_1$: $c_j<x^1_j, x^2_j>$.
For each vertex $v_i$, we create $D$ vertex candidates in $C_2$, denoted as $c^{h}_{i}$ with $1 \leq h \leq D$, each of which corresponds to an edge incident to $v_i$ and has two attributes $c^h_{i}<x^3_{j}, y_i^h>$ or $c^h_{i}<x^4_{j}, y_i^h>$, where $x_j^3$ or $x_j^4$ corresponds to edge $y_i^h$; create $D$ candidates in $C_3$: $d^h_i<z^h_i, y^0_i>$ with $1\leq h \leq D$, and one candidate in $C_4$: $g_{i}<z^0_{i}>$.
Note that, each candidate in $C_1$ represents an edge, each candidate in $C_4$ represents a vertex, and the candidates in $C_2 \cup C_3$ represent the relation between vertices and edges.
Let $C=C_1 \cup C_2 \cup C_3 \cup C_4$.

Next, we construct the following constraints:
\begin{itemize}
    \item $R_1$=$\{x^1_j \rightarrow x^3_j, x^2_j \rightarrow x^4_j|\ e_j\in \mathcal{E}\}$, 
    \item $R_2$=$\{y^0_i \rightarrow z^0_i, \cdots, y^{D}_i \rightarrow z^{D}_i|\ v_i \in \mathcal{V}\}$.
\end{itemize}
Let $R=R_1 \cup R_2$.
It is clear that each candidate has at most two attributes and each attribute occurs at most once in the constraints.
We set the following profit function:
$$
S(c_i)=\left\{
\begin{aligned}
&&&k'+1,& &c_i \in C_1,& \\
&&&0,& &c_i \in C_2\cup C_3,& \\
&&&-1,& &c_i \in C_4.& \\
\end{aligned}
\right.
$$
Let $k:=\frac{5k'(k'-1)}{2}+k'$, $p:=\frac{k'(k'-1)(k'+1)}{2}-k'$.

``$\Longrightarrow$": Suppose that there is a size-$k'$ clique $\mathcal{K}$ in $\mathcal{G}$.
We select $\frac{k'(k'-1)}{2}$ candidates $C'_1 \subset C_1$ and $k'(k'-1)$ candidates $C'_2\subset C_2$, whose corresponding edges are between two vertices in $\mathcal{K}$.
And, we select $k'(k'-1)$ candidates $C'_3 \subset C_3$ which have the same corresponding edges as $C_2'$.
Furthermore, we select $k'$ candidates $C'_4 \subset C_4$ whose corresponding vertices are in $\mathcal{K}$.
The chosen candidate set is $W=C'_1 \cup C'_2 \cup C'_3 \cup C'_4$, $|W|=\frac{5k'(k'-1)}{2}+k'=k$, and the total profit of $W'$ is $\frac{k'(k'-1)(k'+1)}{2}-k'=p$.

Now, we check the satisfaction of the constraints.
The attributes of the selected candidates in $C'_1\cup C'_2$ satisfy the constraints in $R_1$, because these candidates correspond to the $\frac{k'(k'-1)}{2}$ edges with two endpoints in $\mathcal{K}$.
Furthermore, the selected candidates in $C'_3$ and $C'_4$ make sure that the constraints in $R_2$ are all satisfied, since the corresponding edges of $C'_3$ are incident to the $k'$ vertices corresponding to $C'_4$.
In summary, all constraints are satisfied by $W$.
So $W$ is a solution of CECAC.

``$\Longleftarrow$": Suppose that there is a solution $W$ of CECAC.
According to the construction of candidates and constraints, it is clear that when $W$ contains $\sigma$ candidates of $C_1$, there must be at least $2\sigma$ candidates of $C_2$, $2\sigma$ candidates of $C_3$, and some candidates of $C_4$ in $W$.
The selected candidates of $C_4$ correspond to the vertices which are incident to the edges corresponding to the selected candidates from $C_1$.
As a consequence, there are at most $\frac{k'(k'-1)}{2}$ $C_1$-candidates in $W$, since otherwise, the size of $W$ would exceed the committee size $k$.
In addition, to make sure the total profit of selected candidates satisfying $S(W)\geq p$, there must be at least $\frac{k'(k'-1)}{2}$ candidates of $C_1$ in $W$, since even if we do not select any candidates from $C_4$, the total profit of at most $\frac{k'(k'-1)}{2}-1$ $C_1$-candidates is always less than $p$: $(\frac{k'(k'-1)}{2}-1)(k'+1)<p$.
Therefore, there are exactly $\frac{k'(k'-1)}{2}$ candidates of $C_1$ (denoted as $C''_1$), $k'(k'-1)$ candidates of $C_2$, $k'(k'-1)$ candidates of $C_3$, and $k'$ candidates of $C_4$ (denoted as $C''_4$) in $W$.
More specifically, the $\frac{k'(k'-1)}{2}$ edges corresponding to $C''_1$ are only incident to the $k'$ vertices (denoted as $\mathcal{V'}$) corresponding to $C''_4$.
It is clear that $\mathcal{V'}$ is a size-$k'$ clique of $\mathcal{G}$.
\end{proof}

\begin{theorem}
\label{N(a)-NP}
CECAC is NP-hard when each attribute occurs at most twice in the constraints and each candidate has at most one attribute.
\end{theorem}
\begin{proof}
We prove this theorem by giving a reduction from {\scshape Clique} on~$D$-regular graphs.
We construct a CECAC instance $(E=(C,A,R,S,\alpha),k,p)$ from $(\mathcal{G}=(\mathcal{V},\mathcal{E}), k')$ as follows.

Both the candidate set $C$ and the attribute set $A$ consist of three parts: $C=C_1 \cup C_2 \cup C_3$ and $A=A_1 \cup A_2 \cup A_3$.
For each edge $e_j \in \mathcal{E}$, we create an attribute $y_j$ in $A_1$ and a candidate with attribute $y_i$ in $C_1$: $c_j<y_j>$.
For each vertex $v_i \in \mathcal{V}$, we create $D$ attributes $X_i=\bigcup_{1 \leq h \leq D}\{x^h_i\}$ in $A_2$, one attribute $x^0_i$ in $A_3$, $D$ candidates $\bigcup_{1\leq t \leq D}\{c^t_i\}$ in $C_2$, and one candidate $c^0_i$ in $C_3$. 
The $D+1$ candidates of $C_2\cup C_3$ one-to-one correspond to the attributes for $v_i$ in $A_2 \cup A_3$, that is, $c^h_i<x^h_i>$ for $0\leq h \leq D$.
Note that the $D$ candidates in $\bigcup_{1\leq h \leq D}\{c^h_i<x^h_i>\}$ correspond to the $D$ edges incident to the vertex $v_i$ in $\mathcal{G}$.
Note that there are three candidates corresponding to every edge, that is, one candidate in $C_1$ and two candidates in $C_2$.
The $|\mathcal{V}|$ candidates in $C_3$ correspond to the vertices in $\mathcal{V}$. 
Furthermore, we create two sets of constraints as follows:
\begin{itemize}
    \item $R_1$=$\{y_j \rightarrow x^{h_1}_i$, $y_j \rightarrow x^{h_2}_{i'}|\ e_j=\{v_i, v_{i'}\}\in \mathcal{E}\}$;
    \item $R_2$=$\{x^1_i \vee x^2_i \vee \cdots \vee x^D_i \rightarrow x^0_i|\ v_i \in \mathcal{V}\}$.
\end{itemize}
Here, $y_j, x_i^{h_1}$, and $x_{i'}^{h_2}$ are the attributes created for $e_j\in \mathcal{E}$.
Set $R=R_1 \cup R_2$.
It is clear that each attribute occurs at most twice in the constraints and each candidate has only one attribute.
Each attribute $x_i^h$ for $1\leq h\leq D$ occurs once in $R_1$ and once in $R_2$, each attribute $y_j$ occurs twice in $R_1$ but does not occur in $R_2$, and attribute $x_i^0$ occurs only once in $R_2$ and does not occur in $R_1$. 
The profit of each candidate is defined as:
$$
S(c_i)=\left\{
\begin{aligned}
&&1&,& c_i \in C_1, \\
&&0 &,& c_i \in C_2, \\
&&-1&,& c_i \in C_3. \\
\end{aligned}
\right.
$$
Let $k:=\frac{3k'(k'-1)}{2}+k'$ and $p:=\frac{k'(k'-1)}{2}-k'$.
The equivalence of the two instance $\mathcal{G}=(\mathcal{V},\mathcal{E})$ and $(E=(C,A,R,S,\alpha),k,p)$ can be proved in a similar way as the one in theorem~\ref{A(c)-NP}.

``$\Longrightarrow$": Suppose that there is a size-$k'$ clique $\mathcal{K}$ in $\mathcal{G}$.
We selected $\frac{k'(k'-1)}{2}$ candidates $C'_1 \subset C_1$ and $k'(k'-1)$ candidates $C'_2 \subset C_2$, which correspond to the $\frac{k'(k'-1)}{2}$ edges in $\mathcal{K}$.
We also select $k'$ candidates $C_3'\subset C_3$, which correspond to the $k'$ vertices in $\mathcal{K}$.
Let $W=C_1'\cup C_2'\cup C_3'$ be the committee with $S(W)=\frac{k'(k'-1)}{2}-k'=p$ and $|W|=\frac{3k'(k'-1)}{2}+k'=k$.
Since the clique is a complete graph, each pair of vertices is adjacent.
For each selected candidates $c_j\in C'_1$, the two candidates of $C_2$ corresponding to the same edge $e_j$ are selected in $C'_2$.
Therefore, the constraints in $R_1$ are all satisfied.
Meanwhile, all of the $k'$ candidates whose corresponding vertices are incident to the corresponding edges of $C_1'\cup C_2'$ are selected.
This is the reason why all constraints in $R_2$ are satisfied.
Therefore, $W$ is a solution of CECAC.

``$\Longleftarrow$": Suppose that there is a solution $W$ of CECAC, $S(W) \geq p$ and $|W|=k$.
Assume that there are $\frac{k'(k'-1)}{2}+1$ candidates of $C_1$ in $W$, denoted as $C_1'$.
Since the number of the vertices incident to the $\frac{k'(k'-1)}{2}+1$ edges corresponding to $C'_1$ is at least $k'+1$, there must be at least $k'+1$ candidates of $C_3$ in $W$.
According to the construction of $R_1$, there are at least $k'(k'-1)+2$ candidates of $C_2$ in $W$ whose corresponding edges are the same as the corresponding edges of $C_1'$.
In summary, the number of candidates in $W$ is $\frac{k'(k'-1)}{2}+1+k'(k'-1)+2+k'+1>k$, contradicting to that $W$ is a solution of CECAC.
Assume that there are $g$ candidates of $C_1$ in $W$ with $\frac{(k'-1)(k'-2)}{2}\leq g \leq \frac{k'(k'-1)}{2}-1$.
Since the number of the vertices incident to these $g$ edges is at least $k'$, there must be at least $k'$ candidates of $C_3$ in $W$.
So, the total profits of $W$ is $S(W)\leq \frac{k'(k'-1)}{2}-1-k'<p$, again contradicting to that $W$ is a solution of CECAC.
Similarly, we get a contradiction when the number of selected candidates of $C_1$ is less than $\frac{(k'-1)(k'-2)}{2}$.
Therefore, there must be exactly $\frac{k'(k'-1)}{2}$ candidates of $C_1$, $k'(k'-1)$ candidates of $C_2$, and $k'$ candidates of $C_3$ in $W$.
Furthermore, the number of the edges $\mathcal{E'}$ corresponding to the selected candidates of $C_1$ is $\frac{k'(k'-1)}{2}$.
And the vertices $\mathcal{V'}$ corresponding to these $k'$ selected $C_3$-candidates are incident to the edges in $\mathcal{E'}$.
Therefore, $\mathcal{V'}$ is a size-$k'$ clique in $\mathcal{G}$.
\end{proof}

\section{Parameterized Complexity}
In the following, we study the parameterized complexity of CECAC.
The following theorem follows directly from the proof of Theorem~\ref{A(c)-NP} and the W[1]-hardness of {\sc Clique} with respect to the size of clique~\cite{DF99}.
\begin{theorem}
\label{k-p-W[1]}
CECAC is W[1]-hard with respect to the size of committee $k$ and the total profit bound~$p$.
\end{theorem}
Next we consider the parameterized complexity of CECAC with the number of attributes as parameter. 
\begin{theorem}
\label{l-FPT}
CECAC is FPT with respect to the number of attributes~$\ell$.
\end{theorem}
\begin{proof}
Recall that CECAC seeks for a size-$k$ subset~$W\subseteq C$ such that the attributes of the candidates in $W$ satisfy all constraints in $R$ and the total profit of the candidates in $W$ is at least $p$.
We use a length-$\ell$ vector to represent the attributes of each candidate where each component of the vector is $0$ or $1$.
For example, the attributes of $c_i<a_1, a_2>$ can be represented as $<1,1,0,\dots,0>$.
Thus, there are at most~$2^{\ell}$ different vector types.
We call a candidate $c$ is of type $t$ if $c's$ attributes can be represented as a type-$t$ vector, and a type-$t$ vector occurs in $W$ if there is a type-$t$ candidate in $W$.
Let~$T$ denote the set of all vector types, and for each type~$t \in T$, let $C(t)$ be the set of candidates of type $t$.

First, our algorithm enumerates all $2^{2^{\ell}}$ possibilities which vector types occur in the committee.
A vector type occurring in the committee means that the committee contains a candidate, whose attributes can be represented by the vector type.
For each possibility, let $T'$ be the set of vector types occurring in the committee, and $A'=\cup_{t\in T'}A(t)$, where $A(t)$ is the set of attributes whose corresponding components in $t$ are $1$.
If $|T'|>k$ or the attributes in $A'$ do not satisfy all constraints, then continue with the next possibility.
Otherwise, for each type $t\in T'$, add the candidates from $C(t)$ with the highest profit to $W$.
Then, if the number of candidates in $W$ is less than $k$, then add other candidates in~$\cup_{t\in T'}C(t)$ by the decreasing order of their profits to $W$ until there are exactly $k$ candidates in $W$.
Finally, our algorithm compares the total profit $S(W)$ with the total profit bound $p$.
If the total profit of candidates in the committee is less than $p$, we continue with the next possibility.
Otherwise, $W$ is a solution of CECAC.

Next, we prove the correctness and running time of the algorithm.
For each possibility, we can check easily whether the attributes in $A'$ satisfy all constraints in $O(k \times d)$ time.
Afterwards, adding $k$ candidates to $W$ can be done in $O(k \times m)$ time if all constraints are satisfied by the attributes $A'$.
There are at most $2^{2^{\ell}}$ possibilities.
Therefore, the total running time is bounded by $O(2^{2^{\ell}}\times k \times m \times d)$.
Now, we turn to the correctness.
Suppose that there is a solution $W$ for the CECAC instance.
It is clear that the attributes of candidates in $W$ satisfy all constraints, and the number of vector types occurring in $W$ does not exceed $k$ since there are $k$ candidates in $W$.
In addition, the occurrence of the vector types in $W$ is one of the $2^{2^{\ell}}$ possibilities.
For this possibility, our algorithm selects $k$ candidates, denoted as $W'$, satisfying all constraints.
Since our algorithm selects the highest profit candidate for each vector type occurring in the committee and the other candidates with a decreasing order of profits.
It must be $S(W')\geq S(W)\geq p$.
So, $W'$ is also a solution of this CECAC instance.
It means if there is a solution for the CECAC instance, our algorithm can always select a committee solution.
This completes the proof of the correctness of the algorithm.
\end{proof}

The parameterization with the number of constraints leads also to a hardness result. 
\begin{theorem}
\label{d-para-NP}
CECAC is NP-hard even when there is only one constraint.
\end{theorem}
\begin{proof}
We prove the theorem with only one constraint by giving a reduction from {\scshape Independent Set}.
An independent set~$I$ of an undirected graph~$\mathcal{G}$ is a set of pairwisely non-adjacent vertices.
Given a graph~$\mathcal{G}=(\mathcal{V},\mathcal{E})$ and an integer~$k'$, {\scshape Independent Set} asks for a size-$k'$ independent set.
Let~$|\mathcal{V}|=m'$ and $|\mathcal{E}|=n'$ be the number of vertices and edges, respectively.
We construct a CECAC instance $(E=(C,A,R,S,\alpha),k,p)$ from $(\mathcal{G}=(\mathcal{V},\mathcal{E}), k')$ as follows.

For each vertex $v_i \in \mathcal{V}$, we create an attribute $a_i$ and a candidate with this attribute $c_i<a_i>$.
It is $|C|=|A|=|\mathcal{V}|=m'$.
Let $C=\bigcup_{v_i \in \mathcal{V}}\{c_i\}$ and $A=\bigcup_{v_i \in \mathcal{V}}\{a_i\}$.
There is only one constraint
$r: a_1 \vee a_2 \cdots \vee a_{m'} \rightarrow f(e_1)\wedge \cdots \wedge f(e_{n'})$,
where $f(e_j)=(\overline{a_r}\vee \overline{a_s})$ for an edge $e_j=\{v_{r}, v_{s}\}$ in $\mathcal{G}$.
The profit of each candidate is exactly one.
Let $k:=p:=k'$.

Now, we prove the equivalence between the two instances.
Clearly, the committee size and the total profit are always satisfied when selecting $k$ candidates since $k=p=k'$ and the profit of each candidate is $1$.
So, a candidate selected or not is only determined by the constraint $r$.  
According to the construction of $r$, it is clear that $a_1 \vee a_2 \cdots \vee a_{m'}$ is always satisfied if we choose any candidate. 
Furthermore, $f(e_1)\wedge \cdots \wedge f(e_{n'})$ is fulfilled only when all $f(e_j)$ are satisfied for $1\leq j\leq n'$.
Each $f(e_j)$ corresponds to an edge and it can be satisfied if the two candidates corresponding to the endpoints of $e_j$ are not both in $W$. 
It means that the constraint $r$ is satisfied if and only if there exist $k'$ pairwisely non-adjacent vertices.
In other words, there is a solution of CECAC if and only if there is an independent set of size $k'$ in $\mathcal{G}$.
Therefore, CECAC is NP-hard even when there is only one constraint.
\end{proof}
According to this NP-hardness result, we find out that when the structures of the constraint are complex (there are many attributes in a constraint), CECAC is NP-hard even if there is only one constraint.
So, we continue to consider the conditions in the following section where the number of attributes in each constraint is small.

\section{Election with simple attribute constraints}
In this section, we study the complexity of special CECAC problem where the number of attributes occurring in each constraint is limited.
Let $L(r)$ be the number of attributes in a constraint.
Here, we only consider the constraints that cannot be transformed into some other constraints with the same meaning.
For example, if a constraint is $a \vee a' \rightarrow a''$, we consider two constraints $a \rightarrow a''$ and $a' \rightarrow a''$ instead of $a \vee a' \rightarrow a''$. 
From theorem~\ref{A(c)-NP}, we can get the following theorem result straightly that CECAC is NP-hard when each candidate has at most two attributes, each attribute occurs at most once in the constraints, and at most two attributes occurs in a constraint.

\begin{theorem}
\label{A(c)-NP-Simple}
CECAC is NP-hard when each candidate has at most two attributes, each attribute occurs at most once in the constraints, and at most two attributes occurs in a constraint.
\end{theorem}

Similar to the prove in theorem~\ref{N(a)-NP}, we can get the results that 1) CECAC is NP-hard when each attribute occurs at most two times in the constraints, each candidate has at most one attribute, and at most three attributes occurs in a constraint, 2) CECAC is NP-hard when each attribute occurs at most three times in the constraints, each candidate has at most one attribute, and at most two attributes occurs in a constraint, by doing some modifications for each constraints.
For example, given a constraint $r: x_i^1 \vee x_i^2 \vee x_i^3 \vee x_i^4 \rightarrow x_i^0$ in $R_2$ of theorem~\ref{N(a)-NP}, we construct $2$ new attributes $x_{\alpha1}, x_{\alpha2}$ and change the $r$ into the following $3$ constraints:  $x_i^1 \vee x_i^2 \rightarrow x_{\alpha1}$, $x_i^3 \vee x_i^4 \rightarrow x_{\alpha2}$, and $x_{\alpha1} \vee x_{\alpha2} \rightarrow x_i^0$.
The $3$ new constraints have the same meanings as $r$ that if one of the attribute  $x_i^1, x_i^2 , x_i^3$ and $x_i^4$ is chosen, the attribute $ x_i^0$ would also be chosen.
So, we can get the NP-hardness result in the following theorem.
\begin{theorem}
\label{N(a)-NP-Simple}
CECAC is NP-hard when each attribute occurs at most two times in the constraints, each candidate has at most one attribute, and at most three attributes occurs in a constraint.
\end{theorem}

Similarly, we can construct $3$ new attributes $x_{\alpha3}$, $x_{\alpha4}$, $x_{\alpha5}$ and change $r$ into the following $7$ constraints: $x_i^1\rightarrow x_{\alpha3}$, $x_i^2\rightarrow x_{\alpha3}$, $x_i^3\rightarrow x_{\alpha4}$, $x_i^4\rightarrow x_{\alpha4}$, $x_{\alpha3}\rightarrow x_{\alpha5}$, $x_{\alpha4}\rightarrow x_{\alpha5}$, $x_{\alpha5}\rightarrow x_i^0$.
The $7$ new constraints also means that if one of the attribute  $x_i^1, x_i^2 , x_i^3$ and $x_i^4$ is chosen, the attribute $x_i^0$ would also be chosen.
Therefore, according to theorem~\ref{N(a)-NP}, we can get the following theorem directly.

\begin{theorem}
\label{N-NP-Simple}
CECAC is NP-hard when each attribute occurs at most three times in the constraints, each candidate has at most one attribute, and at most two attributes occurs in a constraint.
\end{theorem}

In the following, we present a polynomial-time algorithm to show that when each attribute occurs at most two times in the constraints, each candidate has at most one attribute, and at most two attributes occurs in a constraint, CECAC is polynomial-time solvable.

\begin{theorem}
\label{A(c)-N(a)-p-simple}
When the following three conditions are fulfilled, CECAC is solvable in polynomial time:\\
1) each candidate has at most one attribute, $\forall\ c \in C: |\alpha(c)|\leq 1$;\\
2) each attribute occurs at most twice in the constraints, $\forall\ a\in A: N(a)\leq 2$;\\
3) at most two attributes occurs in a constraint, $\forall r\in R: L(r)\leq 2$.
\end{theorem}
\begin{proof}
Let $(E=(C,A,R,S, \alpha), k, p)$ be a CECAC instance where $|C|=m$, $|A|=\ell$, and $|R|=d$.
Similar to the algorithm~\ref{A(c)-N(a)-p}, we first divide the candidates into $C^+$ and $C^-$: add a candidate $c_i<a_i>$ to $C^+$ if $a_i$ or $\overline{a_i}$ occurs in a constraint; otherwise, add $c_i<a_i>$ to $C^-$.
Our algorithm chooses some candidates from $C^+$ to satisfy all constraints firstly, and then, add or replace some chosen candidates by the candidates in $C^-$ to ensure the total profits at least being $p$.

Our algorithm consists of the following steps.
Firstly, we divide the constraints into different constraint \emph{clusters} where each two constraints from different clusters do not share a common attribute.
Secondly, we analyse the constraints in each cluster and show the logical relationships of these constraints.
Finally, we use a dynamic approach to find out the committee where each constraints are satisfied.
In the following, we show the details of our algorithm.

At first, according to the existence of each attribute in constraints, we divide constraints into different cluster.
Precisely, if attribute $a$ or $\overline{a}$ occurs in a constraint $r_i$, then all constraints containing $a$ or $\overline{a}$ are in the same constraint cluster with $r_i$.
It is clearly that each two clusters do not share a common attribute, and an attribute $a$ chosen or not has no influence on the choose of attribute $a'$ if they are in different clusters.
Therefore, we can analyse the constraints in a set individually, and combine the results from all clusters finally.

For each cluster, we first combine all constraints as follows.\\
\indent
\textcircled{1} If two constraints are $x \rightarrow x'$ and $x' \rightarrow x''$ where $x$, $x'$ and $x''$ are of the form $a$ or $\overline{a} \in A$, then we combine them together as $x\rightarrow x' \rightarrow x''$.
We call $x \rightarrow x' \rightarrow \cdots \rightarrow x''$ as a \emph{string}, denoted as $S_{x,x''}$.
Let $x$ and $x'$ be the first and last attributes of string $S_{x,x'}$, respectively.\\
\indent
\textcircled{2} For each string $S_{x,x'}$, we construct a vertex $v_{x,x'}$.
For each two vertices $v_{x,x'}$ and $v_{x'',x'''}$, if $x'=\overline{x''}$, we construct a directed edge from $v_{x,x'}$ to $v_{x'',x'''}$.
And, if two vertices $v_{x,x'}$ and $v_{x'',x'''}$ satisfy $x\in \{x'', \overline{x''}\}$ or $x' \in \{x''', \overline{x'''}\}$, we combine the two vertices together as a new vertex.

Here, we call attribute $a$ is \emph{fulfilled} if a candidate owns the attribute $a$ in the committee, and $\overline{a}$ is \emph{fulfilled} if all candidates do not own the attribute $a$ in the committee.
Since each attribute occurs at most twice in all constraints, all constraints in a cluster turn into a directed graph where the degree of each vertices is at most 2.
In the following, we show a dynamic approach to select $k$ candidates satisfying all constraints.

For each vertex $v_{x,x'}$: 
\begin{itemize}
    \item let $A^+(S_{x,x'}, j)$ be the set of attributes we need to choose to make the first $j$ attributes unfulfilled and the other $|S_{x,x'}|-j$ attributes fulfilled in $S_{x,x'}$; and $A^-(S_{x,x'}, j)$ be the set of attributes we cannot choose to make the first $j$ attributes unfulfilled and the other $|S_{x,x'}|-j$ attributes fulfilled in $S_{x,x'}$.
    \item let $C(A^+(S_{x,x'}, j))$ be the set of candidates whose attribute is in $A^+(S_{x,x'}, j)$; $C^+(A^+(S_{x,x'}, j))$ be the candidate set that containing a candidate, which owns the highest profit and attribute $a$, for each attribute $a\in A^+(S_{x,x'}, j)$; and $C^-(A^+\\
    (S_{x,x'}))=C(A^+(S_{x,x'}))\setminus{C^+(A^+(S_{x,x'}))}$.
\end{itemize}
We denote $H(i,C)$ being the subset of size-$i$ candidates of $C$ with the highest profits.

For each vertex $v_{x,x'}$, we consider all conditions where all constraints are satisfied for string $S_{x,x'}$, and compute three multisets $V=\bigcup V_i^j(v_{x,x'})$, $N=\bigcup N_i^j(v_{x,x'})$ and $P=\bigcup P_i^j(v_{x,x'})$ with $0\leq i\leq k$ and $0\leq j\leq |S_{x,x'}|$ for each condition.
Hereby, $N_i^j(v_{x,x'})$ and $P_i^j(v_{x,x'})$ are used to record the candidate set where candidates cannot be replaced by the candidates in $C^-$ and the candidate set where candidates can be replaced by the candidates in $C^-$ with selecting exactly $i$ candidates and first $j$ attributes unfulfilled for $S_{x,x'}$, respectively.
The set $V_i^j(v)$ is set one-to-one corresponding to $N_i^j(v)$ and $P_i^j(v)$.

Next, we first initialize the sets $V, N, P$ of each vertex, and then show the dynamic approach to make all constraints satisfied for all vertices in a cluster.
For a vertex $v_{x,x'}$:\\
\indent
1) when $x=x'$, we consider only two conditions that all attributes are fulfilled and no attributes are fulfilled;\\
\indent
2) when $x=\overline{x'}$, we consider $|S_{x,x'}|$ conditions with first $j$ attributes not fulfilled, $1\leq j\leq |S_{x,x'}|$;\\
\indent
3) when $x\notin \{x',\overline{x'}\}$, we consider $|S_{x,x'}|+1$ conditions with first $j$ attributes not fulfilled, $0\leq j\leq |S_{x,x'}|$.\\
\indent
For each of the condition, the elements in $N$, $P$ and $V$ are initialized as follows:
\begin{itemize}
    \item If $|C^+(A^+(S_{x,x'}, j))|\leq i\leq |C(A^+(S_{x,x'}, j))|$, we set $N_i^j(v_{x,x'})=C^+(A^+\\
    (S_{x,x'}, j))$, $P_i^j(v_{x,x'})=H(i-|A^+(S_{x,x'}, j)|, C^-(A^+(S_{x,x'}, j)))$, $V_i^j(v_{x,x'})=S(N_i^j(v_{x,x'})\cup P_i^j(v_{x,x'}))$;
    \item Otherwise, we set $N_i^j(v_{x,x'})=P_i^j(v_{x,x'})=\emptyset$, $V_i^j(v_{x,x'})=-\infty$.
\end{itemize}

Now, we show the details of how to satisfy all constraints of all vertices.
For each two vertices $v_{x,x'}$ and $v_{\overline{x'},x''}$, we do the following update operations for the sets $N$, $P$, and $V$ of $v_{\overline{x'},x''}$ from $v_{x,x'}$ as follows:
 \begin{equation*}
    \begin{split}
        V_{t+t'}^0(v_{\overline{x'},x''})&=\max\{V_{t}^{|S_{x,x'}|}(v_{x,x'})+V_{t'}^{0}(v_{\overline{x'},x''})\},\\
        V_{t+t'}^j(v_{\overline{x'},x''})&=\max\{V_{t}^{t''}(v_{x,x'})+V_{t'}^{j}(v_{\overline{x'},x''})\},
        \end{split}
    \end{equation*}
where $1\leq j\leq |S_{\overline{x'},x''}|$ and $0\leq t, t'\leq k, 0\leq t''\le |S_{x,x'}|$.

The update operations always keep the consistency of variable assignment, especially for the last attribute in $S_{x,x'}$ and first attribute in $S_{\overline{x'},x''}$.
\\
\indent
(1) The graph is a vertex $v$. 
The output of this cluster is the sets $N, P$ and $V$ of $v$ from the initialization.
\\
\indent
(2) The graph is a directed circle.
We start the above update operations at any vertex $v$, and do once update operation for each vertex.
The output of this cluster is the sets $N, P$ and $V$ of vertex $v$.
\\
\indent
(3) The graph is neither a vertex nor a directed circle.
We start the update operations from the vertices with the in-degree being 0 and end at the vertices with out-degree being 0.
Without loss generality, let $r$ be the number of sets $V$ got from the graph, $1\leq r \leq d$.
We can get the sets $N$, $P$, and $V$ of the cluster by combining each two of the $V$ sets from the out-degree being 0 vertices:
\begin{equation*}
    \begin{split}
        V_{i}=\max\{V_{t1}^{t1'}(v_{x1,x1'})+V_{t2}^{t2'}(v_{x2,x2'}\}
    \end{split}
\end{equation*}
where $i=t1+t2$ and $0\leq t1,t2 \leq k$, $0\leq t1'\leq |S_{x1,x1'}|+1$, $0\leq t2'\leq |S_{x2,x2'}|+1$.
The sets $N_i$ and $P_i$ are set corresponding to the $V_i$.

Similarly, we can get the sets $N$, $P$, and $V$ of each cluster. 
And then, we can combine all $V_i$ of each cluster together, and get the sets $N$, $P$ and $V$ which corresponding to all constraints in $R$.
Finally, we add some candidates in $C^-$ into the committee to make sure that there are exactly $k$ candidates, or use the candidates in $C^-$ to replace some candidates in $P$ to maximize the total profits.
We decide whether there is a solution to the CECAC instance by comparing $V_k$ with the total profit bound $p$.

Next, we prove the correctness and running time of the algorithm.
Concerning the running time, we divide the constraints into different sets and combines the constraints into strings.
These operations can be done in $O(d^3)$ times.
For each set, the directed graph corresponding to strings can be constructed in $O(d^3)$ times since there are at most $d$ strings.
For each vertices $S_{x,x'}$, we consider at most $|S_{x,x'}|$ conditions and calculate 3 values for each condition where each value can be computed in $O(k\times d)$ times.
For the dynamic process, each two connected vertices will do an update operation including at most $O(k^4)$ comparisons, and there are at most $d$ such connected vertices.
So, the update operations can be done in $O(k^4\times d)$ time.
Next, our method will consider all $N$, $P$, and $V$ sets from each out-degree being $0$ vertices and get final $N$, $P$, $V$ sets of the graph.
Here, we can consider each two $V$ sets at a time and do at most $d$ times.
For each time, we will do at most $O(k^2)$ times comparisons.
So, the dynamic process can be done in $O(k^4\times d)$ times.
There are at most $d$ sets to analyse and combining the results from the $d$ sets can be done in $O(d\times k^2)$ times.
Finally, adding the candidates $C^-$ into the committee or replacing some candidates in $P$ by the candidates in $C^-$ can be done in $O(m)$ times. 
In summary, the total running time is bounded by $O(k^4\times d^3\times m)$. 

Now, we turn to the correctness.
Obviously, dividing constraints into different sets and considering each set individually have no affect on the output of the election, since the each two sets do not share a common attribute.
According to Conditions 1), for the candidates with the attribute $a$, if we need to choose the attribute $a$, then the candidate with the highest profit must be chosen in $N$, and other candidates can be chosen in $P$ or not.
Otherwise, all of candidates with attribute $a$ can not be chosen.
According to Condition 2) and 3) of the theorem, the degree of each vertices in directed graph is at most $2$.
So, the directed graph must be a circle or has some in-degree being $0$ or out-degree being $0$ vertices.
The update operations can be done from each vertex when the graph is a circle, or start from the in-degree being $0$ vertices and end at the out-degree being $0$ vertices.
For each vertices, the initialization is obviously correct since $C^+(A^+(S_{x,x'}))$ are all the candidates that must be chosen and $C(A^+(S_{x,x'}))$ are all the candidates that we can choose.
For the update operations, we consider all conditions that choosing $i$ candidates including exactly $j$ candidates which can be replaced.
Meanwhile, we also consider the uniformity between the last attribute of the first string and the first attribute of the secondary string.
So, the update operations are correct.
Next, we consider all $O(k^2)$ comparisons of $N$, $P$, and $V$ sets for each two out-degree being $0$ vertices.
Therefore, we can get $N$, $P$, and $V$ sets of each constraint sets correctly. 
Similarly, we can also correctly combine all of these $N$, $P$, and $V$ from each constraint sets to get the $N$, $P$, and $V$ sets corresponding to $R$.
Finally, we add some candidates in $C^-$ into $P$, or use the candidates to replace the candidates in $P$.
In this way, we can make sure that the committee includes exactly $k$ candidates and is with the highest profits.
Therefore, if $V(k)<p$, there is no solution;
otherwise, the candidates in $N$ and $P$ corresponding to $V(k)$ is a solution which satisfies all constraints and the total profit is at least $p$.
This completes the proof of the correctness of the algorithm and thus, the proof of the theorem.
\end{proof}

\section{Conclusion}
In this paper, we investigate the problem of electing a committee, which is required to satisfy
some constraints given as logical expressions of the candidate attributes. 
We provide a dichotomy concerning its classical complexity with respect to the maximal number of attributes of a candidate and the maximal occurrence of attributes in the constraints. 
More precisely, it is polynomial-time solvable if each candidate has exactly one attribute and each attribute occurs once in the constraints; otherwise, it becomes NP-hard.
We also derive a collection of parameterized complexity results.
The problem is FPT with respect to the number of attributes, and is W[1]-hard or para-NP-hard with the size of committee, the total profit bound, or the number of constraints as parameter.

For the future work, it is interesting to examine the relation between the complexity and other structural parameters of constraints, such as the maximal length of constraints.
Another promising direction for future work might be the approximability of the optimization version of the problem.
There are two possible objective functions, the maximal number of satisfied constraints with certain total profit bound and the maximal total profit with all constraints being satisfied.
\clearpage


\end{document}